\documentclass[12pt,reqno]{amsart}

\newcommand\version{April 23, 2019}

\usepackage{amsmath,amsfonts,amsthm,amssymb,amsxtra}
\usepackage{bbm} 

\newtheorem{theorem}{Theorem}
\newtheorem{proposition}[theorem]{Proposition}
\newtheorem{lemma}[theorem]{Lemma}
\newtheorem{corollary}[theorem]{Corollary}

\theoremstyle{definition}

\theoremstyle{remark}

\renewcommand{\epsilon}{\varepsilon}

\newcommand{\N}{\mathbb{N}}

\renewcommand{\phi}{\varphi}
\newcommand{\R}{\mathbb{R}}

\newcommand{\Sph}{\mathbb{S}}

\DeclareMathOperator{\per}{Per}
\DeclareMathOperator{\ran}{ran}

\DeclareMathOperator{\spa}{span}

\begin{document}

\title[Non-spherical equilibrium shapes --- \version]{Non-spherical equilibrium shapes\\ in the liquid drop model}

\author{Rupert L. Frank}
\address[R. Frank]{Mathematisches Institut, Ludwig-Maximilans Universit\"at M\"unchen, The\-resienstr. 39, 80333 M\"unchen, Germany, and Department of Mathematics, California Institute of Technology, Pasadena, CA 91125, USA}
\email{r.frank@lmu.de}

\begin{abstract}
We prove the existence of a family of volume-constrained critical points of the liquid drop functional, which are cylindrically but not spherically symmetric. This family bifurcates from the ball and exchanges stability with it. We justify a formula of Bohr and Wheeler for the energy of these sets.
\end{abstract}

\renewcommand{\thefootnote}{${}$} \footnotetext{\copyright\, 2019 by the author. This paper may be reproduced, in its entirety, for non-commercial purposes.\\
The author thanks T. K\"onig for comments on an early version of the manuscript. Partial support through US National Science Foundation grant DMS-1363432 is acknowledged.}

\maketitle


\section{Introduction and main result}

\subsection{Introduction}

Gamow's liquid drop model \cite{Ga} is a classical model of a nucleus which, despite its simplicity, is believed to make qualitatively correct predictions. Recently, it has received a lot of interest in the mathematics literature, see, for instance, \cite{LuOt,KnMu,Ju,BoCr,FrLi,FrKiNa} as well as the review \cite{ChMuTo} and the references therein.

In the liquid drop model, nuclei are considered as arbitrary measurable sets $E\subset\R^3$ of positive and finite measure. The nucleon density is assumed to be constant and therefore the measure $|E|$ is interpreted, in suitable units, as the nucleon number. The corresponding energy, in dimensionless units, is given by the functional
$$
\mathcal I[E] := \per E + D[E] \,,
$$
where $\per E$ denotes the perimeter in the sense of geometric measure theory (equal to the surface area for sufficiently regular sets) and where
$$
D[E] := \frac12 \iint_{E\times E} \frac{dx\,dy}{|x-y|}
$$
denotes the Coulomb repulsion between the protons. The ground state energy at nucleon number $A$ (considered here as a continuous positive parameter) is
\begin{equation}
\label{eq:groundstateenergy}
\inf\left\{ \mathcal I[E] :\ |E|=A \right\}.
\end{equation}

It is widely believed, but not proved, that for $A\leq A_{c}= 5(2-2^{2/3})/(2^{2/3}-1)\approx 3.512$ the infimum in \eqref{eq:groundstateenergy} is attained precisely when $E$ is a ball and that for $A>A_c$ the infimum is not attained. The value of $A_{c}$ is determined by the equality of the energy of a single ball with that of two balls of equal radii which are infinitely far apart. This conjecture appears explicitly, for instance, in \cite{ChPe2}. What is rigorously known is that the infimum in \eqref{eq:groundstateenergy} is attained at balls when $A$ is small \cite{KnMu,Ju,BoCr} and that the infimum is not attained if $A>8$ \cite{FrKiNa}, see also \cite{LuOt,KnMu}.

In this paper we are concerned not with solutions of the minimization problem \eqref{eq:groundstateenergy}, but more generally with volume-constrained critical points of $\mathcal I$. For any $A>0$, balls of volume $A$ are volume-constrained critical points of $\mathcal I$. They are stable against local perturbations (in the sense of having a positive semi-definite second variation when restricted to variations with mean zero) if and only if $A\leq 10$. This is remarkable since $10>A_c$. This computation is well-known in the physics literature and appears, for instance, in \cite{BoCr}. Recently, it was shown in \cite{Ju1} that for any $M>0$ there is an $A_M>0$ such that balls are the only stable volume-constrained critical points $E$ of $\mathcal I$ with $|E|<A_M$ and $\per E\leq M |E|^{2/3}$.

In this paper we are concerned with volume-constrained critical points for non-small volumes $A$. Our main result is the existence of a smooth family of non-spherical volume-constrained critical points with volumes close to 10. This family bifurcates from the ball of volume $A=10$, where the ball loses stability. The sets that we construct are cylindrically symmetric and change from prolate (that is, football shaped) for volumes below $10$ to oblate (that is, pancake shaped) for volumes above $10$. The energy of these sets is above that for balls of the same volume for volumes below 10 and below it for volumes above 10. Moreover, at volume $10$ an exchange of stability takes place between balls and the new, non-spherical sets in the sense that the latter are stable for volumes above 10 and unstable for volumes below 10.

The sets whose existence we prove have been studied before by Bohr and Wheeler \cite[Section II]{BoWh}. They argue that these sets appear as an intermediate state when a ball decays into two balls. More precisely, they consider volume-preserving deformations of a ball with mass between $A_c$ and $10$. Since such balls are local, but not global minimizers, the energy along a deformation path first increases and then decreases. Of obvious physical interest is the difference between the maximal energy along this deformation path and the initial energy and, in particular, the infimum of this quantity with respect to all volume-preserving deformation paths. This min-max value, if it exists, should correspond to saddle-point solutions, namely the ones considered in this paper.

Our construction is different from the one proposed by Bohr and Wheeler. However, due to the uniqueness of our construction, if there are saddle-point solutions as suggested by Bohr and Wheeler which are close to a ball of volume 10, then these solutions necessarily coincide with our solutions. In this sense, our work rigorously justifies the Bohr--Wheeler formula for the leading order behavior of the deformation energy for volumes close to 10.

Our analysis is purely local around volume 10. Different arguments would be required to understand the global behavior of the bifurcation branch. In particular, the work of Bohr and Wheeler suggest that the branch continues to arbitrary small volumes and that the sets converge to two touching balls as the volume tends to zero. To prove this is an open problem. Note that the existence of the bifurcation branch with arbitrarily small volumes is consistent with the result from \cite{Ju} mentioned above, since the conjectured branch is believed to lose its stability at a certain volume; see \cite{BuGa} and also \cite[Figure 3]{Ni}.

In this paper we will deduce the existence of these sets from the bifurcation theorem of Crandall and Rabinowitz \cite{CrRa}, after having identified star-shaped, volume-constrained critical points of $\mathcal I$ with solutions of a certain quasi-linear partial differential equation on $\Sph^2$. Our construction bears some similarity with those in \cite{ReWe,CaFaWe1,CaFaWe2}, although these works deal with seemingly quite different problems.

We finally point out that there is a version of the liquid drop model for nuclear matter with a neutralizing background. A mathematically similar model appears in the theory of diblock polymers; see, e.g., \cite{AlChOt,ChPe,ChPe2,KnMuNo, EmFrKo,FrLi2} and references therein. It would be interesting to understand bifurcations from spherical, cylindrical and lamellar shapes in these models. The paper \cite{Fa} is a first step in this direction, but with Yukawa instead of Coulomb interaction.


\subsection{Main results}

We will consider sets of the form
$$
\Omega_\phi := \left\{ x\in\R^3 :\ |x|< \phi(x/|x|) \right\},
$$
where $\phi:\Sph^2\to\R$ is continuous and non-negative. The following lemma says that these sets are volume-constrained critical points of the functional $\mathcal I$ if and only if $\phi$ solves a certain quasi-linear PDE.

\begin{lemma}\label{equivalence}
Let $\phi:\Sph^2\to \R$ be positive and Lipschitz. Then $\phi$ satisfies
$$
\frac{d}{dt}\Big|_{t=0} \,\mathcal I[ |\Omega_{\phi}|^{1/3} |\Omega_{\phi+tu}|^{-1/3} \Omega_{\phi+tu}] = 0
$$
for every Lipschitz $u : \Sph^2\to\R$ if and only if it satisfies
\begin{equation}
\label{eq:nonlinear2}
-\nabla\cdot \frac{\nabla\phi}{\phi \sqrt{\phi^2 + (\nabla\phi)^2}} + \frac{3}{\sqrt{\phi^2+(\nabla\phi)^2}} - \frac{\sqrt{\phi^2+(\nabla\phi)^2}}{\phi^2} + \int_{\Omega_\phi} \frac{dy}{|\phi(\omega)\omega - y|}  = \mu
\quad\text{on}\ \Sph^2
\end{equation}
in the weak sense with
\begin{equation}
\label{eq:muvirial}
\mu = |\Omega_\phi|^{-1} \left( \frac23 \per\Omega_\phi + \frac53 D[\Omega_\phi] \right).
\end{equation}
\end{lemma}

In \eqref{eq:nonlinear2}, $\nabla$ denotes the gradient on $\Sph^2$ and $\nabla\cdot$ the associated divergence.

By a simple computation we see that for any $R>0$, $\phi\equiv R$ solves \eqref{eq:nonlinear2} with $\mu = 2R^{-1} + \frac{4\pi}3 R^2$. We are looking for solutions which are small perturbations of these constant solutions.

For a positive function $\phi\in C^2(\Sph^2)$ we denote the function on the left side of \eqref{eq:nonlinear2} by $F(\phi)$, that is,
$$
F(\phi) = -\nabla\cdot \frac{\nabla\phi}{\phi\sqrt{\phi^2+(\nabla\phi)^2}} + \frac{3}{\sqrt{\phi^2+(\nabla\phi)^2}} - \frac{\sqrt{\phi^2+(\nabla\phi)^2}}{\phi^2} + \int_{\Omega_\phi} \frac{dy}{|\phi(\omega)\omega-y|} \,.
$$
Moreover, we define
\begin{equation}
\label{eq:defrstarp}
R_* = \left( \frac{30}{4\pi} \right)^{1/3}
\qquad\text{and}\qquad
P(\omega) = \frac{3\omega_3^2-1}{2} \,.
\end{equation}

Our main result reads as follows.

\begin{theorem}\label{main}
Let $0<\alpha<1$. Then there are $\epsilon>0$ and $C^\infty$ curves
$$
R:(-\epsilon,\epsilon)\to (0,\infty),\ s\mapsto R_s \,,\qquad
\chi: (-\epsilon,\epsilon)\to C^{2,\alpha}(\Sph^2),\ s\mapsto \chi_s \,,
$$
with the following properties:
\begin{enumerate}
\item[(a)] $R_s\to R_*$ and $s^{-1} \chi_s \to P$ in $C^{2,\alpha}(\Sph^2)$ as $s\to 0$.
\item[(b)] For all $s\in(-\epsilon,\epsilon)$, $\chi_s$ depends only on $|\omega_3|$ and satisfies
\begin{equation}
\label{eq:orthogonality}
\int_{\Sph^2} P(\omega)(\chi_s(\omega)-sP(\omega))\,d\omega = 0 \,.
\end{equation}
\item[(c)] For all $s\in(-\epsilon,\epsilon)$, $F(R_s+\chi_s)=F(R_s)$.
\end{enumerate}
Moreover, there is a neighborhood $U$ of $(R_*,0)$ in $(0,\infty)\times \{\chi\in C^{2,\alpha}(\Sph^2): \ \chi \ \text{depends}$ $\text{only on}\ |\omega_3|\}$ such that if $(R,\chi)\in U$ satisfies $F(R+\chi)=F(R)$, then either $\chi\equiv 0$ or $(R,\chi)=(R_s,\chi_s)$ for some $s\in(-\epsilon,\epsilon)$.
\end{theorem}

Because of Lemma \ref{equivalence}, item (c) in Theorem \ref{main} means that for all $s\in(-\epsilon,\epsilon)$, $\Omega_{R_s+\chi_s}$ is a volume-constrained critical point of $\mathcal I$.

Our remaining results concern properties of the sets $\Omega_{R_s+\chi_s}$. They rely on the following theorem which computes $R_s$ and $\chi_s$ to next order. We define
\begin{equation}
\label{eq:defq}
Q(\omega) = R_*^{-1} \frac{3^3}{17\cdot 35} (35\omega_3^4 - 30\omega_3^2 + 3) - R_*^{-1} \frac{2}{15} \,.
\end{equation}

\begin{theorem}\label{secondorder}
As $s\to 0$,
\begin{equation}
\label{eq:secondorderr}
R_s = R_* -\frac 17 s + \mathcal O(s^2)
\end{equation}
and, in $C^{2,\alpha}(\Sph^2)$ for any $\alpha<1$,
$$
\chi_s = s P + s^2 Q + \mathcal O(s^3) \,.
$$
\end{theorem}

Expansion \eqref{eq:secondorderr} implies that the bifurcation is transcritical (see, for instance, \cite{Ki}).

A first consequence of this theorem concerns the volume of the sets $\Omega_{R_s+\chi_s}$.

\begin{corollary}\label{volume}
As $s\to 0$,
\begin{align*}
|\Omega_{R_s+\chi_s}| = 10 - s \frac{30}{7} R_*^{-1} + \mathcal O(s^2) \,.
\end{align*}
\end{corollary}

In particular, $|\Omega_{R_s+\chi_s}|$ takes values both above and below $|\Omega_{R_*}|=10$.

Theorem \ref{secondorder} also has consequences concerning stability properties. For $R$ close to $R_*$ the linearization of \eqref{eq:nonlinear2} around the constant solution $\phi\equiv R$, when restricted to functions depending only on $|\omega_3|$, has a unique eigenvalue close to zero and this eigenvalue is given by $4-(8\pi/15)R^3$. This will be shown in Proposition \ref{linearization} below. In particular, the eigenvalue is positive for $R<R_*$ and negative for $R>R_*$.

\begin{corollary}\label{eigenvalue}
For $s$ close to zero, the linearization of \eqref{eq:nonlinear2} around $R_s+\chi_s$, when restricted to functions depending only on $|\omega_3|$, has a unique eigenvalue $\lambda(s)$ close to zero. Moreover, as $s\to 0$,
$$
\lambda(s) = - \frac{12}{7} R_*^{-1} s + \mathcal O(s^2) \,.
$$
\end{corollary}

In particular, the eigenvalue $\lambda(s)$ is negative for $s>0$ (that is, $|\Omega_{R_s+\chi_s}|<10$) and positive for $s<0$ (that is, $|\Omega_{R_s+\chi_s}|>10$). This means that an exchange of stability occurs at the bifurcation point.

Finally, we compare the energy of $\Omega_{R_s+\chi_s}$ with that of a ball with the same volume.

\begin{theorem}\label{energy}
If $\rho_s>0$ is defined by
$$
|\Omega_{R_s+\chi_s}| = |B_{\rho_s}| \,,
$$
then, as $s\to 0$,
\begin{equation}
\label{eq:bohrwheeler}
\mathcal I[\Omega_{R_s+\chi_s}] = \mathcal I[B_{\rho_s}] + \frac{8\pi}{35} R_*^{-1} s^3 + \mathcal O(s^4) \,.
\end{equation}
\end{theorem}

In particular, the energy of $\Omega_{R_s+\chi_s}$ is above that of the ball of the same volume for $s>0$ (that is, $|\Omega_{R_x+\chi_s}|<10$) and below it for $s<0$ (that is, $|\Omega_{R_x+\chi_s}|>10$).

We claim that formula \eqref{eq:bohrwheeler} coincides with the leading order term in the Bohr--Wheeler formula \cite[(24)]{BoWh}. Their formula is stated in terms of
$$
x= \frac{4\pi}{30} \rho_s^3
\qquad\text{and}\qquad
f(x) = \frac{\mathcal I[\Omega_{R_s+\chi_s}] - \mathcal I[B_{\rho_s}]}{\per B_{\rho_s}}
$$
and reads, to leading order,
$$
f(x) \sim \frac{98}{135} (1-x)^3
\qquad\text{as}\ x\to 1 \,;
$$
see also \cite{FrMe,Sw} for more explanations. According to Corollary \ref{volume}, we have $1-x \sim (3/7)R_*^{-1}s$ and therefore
$$
\frac{98}{135}(1-x)^3\per B_{\rho_s} \sim \frac{98}{135} \frac{3^3}{7^3} R_*^{-3} (4\pi) R_*^2 s^3 = \frac{8\pi}{35} R_*^{-1} s^3 \,,
$$
which gives the equivalence of our and their formula. Finally, we note in passing that our formula for $Q$ in Theorem \ref{secondorder} does not coincide with the corresponding formula \cite[(23)]{BoWh}. However, there is a misprint in the latter formula, as observed in \cite{PrKn2}. Our formula for $Q$ coincides with what is obtained from \cite{PrKn1}, see also \cite{Sw}.


\subsection{Ingredients in the proof}

Let us explain the strategy of the proofs of the results mentioned in the previous subsection. We will defer the proofs of various technical assertions to the following sections. Let us fix $0<\alpha<1$ and consider
$$
\mathcal O = \left\{ (R,\chi)\in (0,\infty)\times C^{2,\alpha}(\Sph^2) :\ \inf_{\Sph^2}\chi > -R \right\},
$$
which is open in $(0,\infty)\times C^{2,\alpha}(\Sph^2)$. For $(R,\chi)\in\mathcal O$ we define
$$
\Phi(R,\chi) := R^2 \left( F(R+\chi)-F(R) \right).
$$

\begin{proposition}\label{smooth}
The map $\Phi:\mathcal O\to C^{0,\alpha}(\Sph^2)$ is $C^\infty$.
\end{proposition}

The following lemma computes the first derivative of $\Phi$ with respect to $\chi$ at $(R,0)$.

\begin{proposition}\label{linearization}
For any $R>0$,
$$
D_\chi\Phi(R,0)[u](\omega) = -\Delta u(\omega) - 2u(\omega) + R^3 \left( \int_{\Sph^2} \frac{u(\omega')}{|\omega-\omega'|}\,d\omega' - \frac{4\pi}{3} u(\omega) \right).
$$
This operator commutes with rotations and its eigenvalue on the space of spherical harmonics of degree $\ell\in\N_0$ is
$$
\ell(\ell+1) - 2 - \frac{4\pi}{3} R^3 \left( 1 - \frac{3}{2\ell+1} \right).
$$
\end{proposition}

In the following we work with the subspaces
$$
X = \left\{ \chi\in C^{2,\alpha}(\Sph^2):\ \chi \ \text{depends only on}\ |\omega_3| \right\}
$$
and
$$
Y = \left\{ \chi\in C^{0,\alpha}(\Sph^2):\ \chi \ \text{depends only on}\ |\omega_3| \right\}.
$$
It is easy to see that $\Phi(R,\chi)\in Y$ if $\chi\in X$. We denote by
$$
L_R = D_\chi\Phi(R,0)\Big|_X
$$
the restriction of $D_\chi\Phi(R,0)$ to $X$.

The next lemma clarifies the roles of $R_*$ and $P$ from \eqref{eq:defrstarp}.

\begin{proposition}\label{asscr}
We have
$$
\ker L_{R_*} = \spa \{P\}
\qquad\text{and}\qquad
\ran L_{R_*} = \left\{ \chi\in Y:\ \int_{\Sph^2} P(\omega)\chi(\omega)\,d\omega = 0 \right\}
$$
and
\begin{equation}
\label{eq:LRPderivative}
\frac{d}{dR}\Big|_{R=R_*} L_R P = - 12 R_*^{-1} P \not\in\ran L_{R_*} \,.
\end{equation}
\end{proposition}

For the proof of Theorem \ref{secondorder} we also need the explicit expression for the second derivative of $\Phi$ with respect to $\chi$ at $(R,0)$. This is conveniently stated in terms of the Legendre polynomials
$$
P_2(t) = (3t^2-1)/2
\qquad\text{and}\qquad
P_4(t) = (35 t^4-30t^2+3)/8 \,.
$$
Note that $P(\omega) = P_2(\omega_3)$ and $Q(\omega)=R_*^{-1} ( (6^3/(17\cdot 35)) P_4(\omega_3)- 2/15)$.

\begin{proposition}\label{secondderivative}
One has
\begin{equation}
\label{eq:secondderivative}
\frac12 D_{\chi\chi}^2\Phi(R_*,0)[P,P](\omega) = - 12 R_*^{-1} \left( \frac{12}{35} P_4(\omega_3) + \frac{1}{7} P_2(\omega_3) + \frac{1}{5} \right).
\end{equation}
\end{proposition}


\subsection{Proof of Theorems \ref{main} and \ref{secondorder} and Corollaries \ref{volume} and \ref{eigenvalue}}

We now show how the ingredients from the previous subsection imply our main results.

\begin{proof}[Proof of Theorem \ref{main}]
We will deduce Theorem \ref{main} from the Crandall--Rabinowitz theorem \cite[Theorems 1.7 and 1.18]{CrRa} applied to $\Phi$, considered as a map from $\mathcal O\cap ((0,\infty)\times X)$ to $Y$. The assumptions of that theorem are satisfied by Propositions \ref{smooth}, \ref{linearization} and \ref{asscr}. As the complement of $\ker L_R$ we choose $\{ \chi\in X :\ \int_{\Sph^2} P(\omega) \chi(\omega)\,d\omega = 0 \}$.
\end{proof}

\begin{proof}[Proof of Theorem \ref{secondorder}]
We denote
$$
\alpha = \frac{d}{ds}\Big|_{s=0} R_s \,,
\qquad
\tilde Q = \frac12 \frac{d^2}{ds^2}\Big|_{s=0} \chi_s
$$
and want to show that $\alpha =-1/7$ and $\tilde Q=Q$.
It follows from \cite[Theorem 1.18]{CrRa} (with $n=1$) that
\begin{equation}
\label{eq:equationsecondorder}
\frac{1}{2} D_{\chi \chi}^2\Phi(R_*,0)[P,P] + D_\chi\Phi(0,R_*)[\tilde Q] + \alpha D_RD_\chi\Phi(R_*,0)[P] = 0 \,.
\end{equation}
(Indeed, this follows by differentiating the equation $f(s,R_s,s^{-1}\chi_s-P)=0$ at $s=0$, where $f(s,R,u)=s^{-1}\Phi(R,s(P+u))$ for $s\neq 0$ and $f(0,R,u)=D_\chi\Phi(R,0)[P+u]$.) Inserting \eqref{eq:LRPderivative} and \eqref{eq:secondderivative} into \eqref{eq:equationsecondorder} we obtain
\begin{equation}
\label{eq:equationsecondorder2}
- 12 R_*^{-1} \left( \frac{12}{35} P_4(\omega_3) + \frac{1}{7} P_2(\omega_3) + \frac{1}{5} \right)
+ L_{R_*} \tilde Q(\omega) - \alpha 12 R_*^{-1} P(\omega) = 0 \,.
\end{equation}
We multiply this equation by $P$ and integrate over $\Sph^2$. Using the fact that $L_{R_*}$ is self-adjoint in $L^2(\Sph^2)$ with $L_{R_*}P=0$, as well as the fact that
$$
\int_{\Sph^2} P(\omega) P_4(\omega_3)\,d\omega = 2\pi \int_{-1}^1 P_2(t) P_4(t)\,dt = 0
\ \text{ and}\
\int_{\Sph^2} P(\omega)\,d\omega = 2\pi \int_{-1}^1 P_2(t)\,dt = 0 \,,
$$
we obtain $\alpha=-1/7$, as claimed. Thus, \eqref{eq:equationsecondorder2} becomes
\begin{equation*}
- 12 R_*^{-1} \left( \frac{12}{35} P_4(\omega_3) + \frac{1}{5} \right)
+ L_{R_*} \tilde Q(\omega) = 0 \,.
\end{equation*}
Using the fact that $P_4(\omega_3)$ is a spherical harmonic of degree four, that, by Proposition~\ref{linearization}, $L_{R_*}$ is diagonal in the basis of spherical harmonics and that, by \eqref{eq:orthogonality}, $\int_{\Sph^2} P \tilde Q\,d\omega =0$, we infer that $\tilde Q(\omega) = aP_4(\omega_3) + b$ for some $a,b\in\R$. Using the explicit expression for the eigenvalues of $L_{R_*}$ on spherical harmonics of degrees zero and four from Proposition \ref{linearization}, we find
$$
a = \frac{6^3}{17\cdot 35} R_*^{-1}
\qquad\text{and}\qquad
b=- \frac{2}{15} R_*^{-1} \,,
$$
which shows that, indeed, $\tilde Q = Q$.
\end{proof}

\begin{proof}[Proof of Corollary \ref{volume}]
The claimed expansion for the volume follows easily from the expansion \eqref{eq:secondorderr} of $R_s$ and the fact that $\int P\,d\omega =0$. A more detailed expansion appears in \eqref{eq:expvol} below, so here we omit the details.
\end{proof}

\begin{proof}[Proof of Corollary \ref{eigenvalue}]
Behind the proof is a general argument for transcritical bifurcations, which can be found, for instance, in \cite[Section I.7]{Ki}, but we sketch the argument for the sake of completeness.

The operator in question is the restriction of $D_\chi\Phi(R_s,\chi_s)$ to functions depending only on $|\omega_3|$. The fact that for $s$ close to zero this operator has a single eigenvalue $\lambda(s)$ close to zero follows by continuity from the corresponding fact for the operator $L_{R_*}$. The associated eigenfunction can be chosen of the form $P+w(s)$ with $w(0)=0$. Differentiating the equation
$$
D_\chi\Phi(R_s,\chi_s)[P+w(s)]=\lambda(s)(P+w(s))
$$
with respect to $s$ at $s=0$ gives
$$
D_{\chi\chi}^2\Phi(R_*,0)[P,P] + D_\chi\Phi(R_*,0)[w'(0)] + D_RD_\chi\Phi(R_*,0)[P] R'(0) = \lambda'(0)P \,.
$$
We multiply this equation by $P$ and integrate over $\Sph^2$. Since $D_\chi\Phi(R_*,0)$ is self-adjoint in $L^2(\Sph^2)$ and has $P$ in its kernel, the term involving $w'(0)$ disappears. Using \eqref{eq:secondderivative}, \eqref{eq:secondorderr} and \eqref{eq:LRPderivative}, as well as the same orthogonality relations as in the proof of Theorem \ref{secondorder}, we obtain
$$
-24 R_*^{-1} \frac{1}{7} \int_{\Sph^2} P^2\,d\omega 
+ \frac{12}{7} R_*^{-1} \int_{\Sph^2} P^2\,d\omega 
=\lambda'(0) \int_{\Sph^2} P^2\,d\omega \,.
$$
Thus, $\lambda'(0)=-(12/7)R_*^{-1}$, as claimed.
\end{proof}

This completes our overview over the proofs of our main results. To summarize, we have reduced the proofs of Theorems \ref{main} and \ref{secondorder} and of Corollaries \ref{volume} and \ref{eigenvalue} to the proofs of Propositions \ref{smooth}, \ref{linearization}, \ref{asscr} and \ref{secondderivative}. Those of the first three propositions will be given in Sections \ref{sec:bifurc} and that of the latter proposition in Section \ref{sec:second}.

The remaining two results from the previous subsection, namely, Lemma \ref{equivalence} and Theorem \ref{energy}, follow by expanding the energy  functional to first and third order, respectively. Their proofs will be given in Sections \ref{sec:equivalence} and \ref{sec:energy}, respectively.




\section{The equation for equilibrium shapes}\label{sec:equivalence}

\subsection{Geometric preliminaries}\label{sec:prelim}

We begin by collecting formulas which express quantities built on $\Omega_\phi$ more explicitly in terms of $\phi$. We have
\begin{equation}
\label{eq:volumestarshaped}
|\Omega_\phi| = \int_{\Sph^2} \int_0^{\phi(\omega)} r^2\,dr\,d\omega =\frac13 \int_{\Sph^2} \phi(\omega)^3\,d\omega 
\end{equation}
and
\begin{equation}
\label{eq:coulombstarshaped}
D[\Omega_\phi] = \frac12 \int_{\Sph^2} \int_{\Sph^2} \int_0^{\phi(\omega)} \int_0^{\phi(\omega')} \frac{1}{|r\omega - r'\omega'|} r'^2 \,dr'\,r^2\,dr\,d\omega\,d\omega' \,.
\end{equation}
Moreover, if $\phi$ is Lipschitz, then in the parametrization $x=\phi(\omega)\omega$ the surface measure $d\sigma(x)$ on $\partial\Omega_\phi$ is given by 
$$
\phi(\omega)\sqrt{\phi(\omega)^2 + (\nabla\phi(\omega))^2}\,d\omega = d\sigma(x) \,.
$$
In particular,
\begin{equation}
\label{eq:perimeterstarshaped}
\per \Omega_\phi = \int_{\Sph^2} \phi(\omega) \sqrt{\phi(\omega)^2 + (\nabla\phi(\omega))^2}\,d\omega \,.
\end{equation}
Finally, we recall (see, e.g., \cite[Proposition 4.1]{CaFaWe1}) that the outer unit normal to $\partial\Omega_\phi$ at $x=\phi(\omega)\omega$ is
$$
\nu_x = \frac{\phi(\omega)\omega - \nabla\phi(\omega)}{\sqrt{\phi(\omega)^2 + (\nabla\phi(\omega))^2}} \,.
$$
Using these formulas we will rewrite the volume integral in \eqref{eq:nonlinear2} as a surface integral over $\Sph^2$. We also obtain a corresponding expression for $D[\Omega_\phi]$ which, however, will not be used in this paper.

\begin{lemma}\label{potsurf}
Let $\phi$ be a positive Lipschitz function on $\Sph^2$. Then
$$
\int_{\Omega_\phi} \frac{dy}{|\phi(\omega)\omega-y|} = \frac{1}{2} \int_{\Sph^2} \frac{(\phi(\omega')\omega'-\nabla\phi(\omega'))\cdot(\phi(\omega')\omega'-\phi(\omega)\omega)}{|\phi(\omega)\omega-\phi(\omega')\omega'|} \phi(\omega')\,d\omega'
$$
and
\begin{align*}
D[\Omega_\phi] & = - \frac{1}{4} \iint_{\Sph^2\times\Sph^2} \phi(\omega)\phi(\omega') (\phi(\omega)\omega-\nabla\phi(\omega))\cdot(\phi(\omega')\omega'-\nabla\phi(\omega')) \\
& \qquad\qquad\qquad\times |\phi(\omega)\omega-\phi(\omega')\omega'| \,d\omega \,d\omega'\,.
\end{align*}
\end{lemma}

\begin{proof}
Since $\nabla\cdot (x/|x|) = 2/|x|$ on $\R^3$, we have for any Lipschitz $\Omega\subset\R^3$
$$
\int_\Omega \frac{dy}{|x-y|} = \frac{1}{2} \int_\Omega \nabla_y\cdot \frac{y-x}{|y-x|}\,dy = \frac{1}{2} \int_{\partial\Omega} \nu_y \cdot \frac{y-x}{|y-x|}\,d\sigma(y) \,.
$$
Thus, using $\nabla |x|= x/|x|$ on $\R^3$,
\begin{align*}
D[\Omega] & = \frac12 \int_\Omega \left( \int_\Omega \frac{dy}{|x-y|} \right)dx
= \frac14 \int_{\partial\Omega} \nu_y\cdot \left( \int_\Omega \frac{y-x}{|y-x|}\,dx \right)d\sigma(y) \\
& = - \frac14 \!\int_{\partial\Omega} \nu_y\cdot \!\left( \!\int_\Omega \nabla_x |x-y|\,dx \!\right)d\sigma(y) = - \frac14 \!\int_{\partial\Omega} \nu_y\cdot \!\left( \!\int_{\partial\Omega} \nu_x |x-y|\,d\sigma(x) \!\right)d\sigma(y).
\end{align*}
Inserting in these two formulas, for $\Omega=\Omega_\phi$, the above expressions for $\nu_y$ and $d\sigma(y)$ we obtain the claimed formulas in the lemma.
\end{proof}


\subsection{Derivation of the equation}

We are now in position to give the

\begin{proof}[Proof of Lemma \ref{equivalence}]
We have
$$
\mathcal I[ |\Omega_{\phi}|^{1/3} |\Omega_{\phi+tu}|^{-1/3} \Omega_{\phi+tu}] = \left( \frac{|\Omega_{\phi}|}{|\Omega_{\phi+tu}|}\right)^{2/3} \per\Omega_{\phi+tu} + \left( \frac{|\Omega_{\phi}| }{|\Omega_{\phi+tu}|}\right)^{5/3} D[\Omega_{\phi+tu}] \,.
$$
By straightforward expansions, using \eqref{eq:volumestarshaped} and \eqref{eq:perimeterstarshaped}, we find
$$
|\Omega_{\phi+tu}| = \frac{1}{3} \int_{\Sph^2} (\phi+tu)^3\,d\omega = |\Omega_\phi| + t \int_{\Sph^2} \phi^2 u \,d\omega + o(t) \,.
$$
and
$$
\per \Omega_{\phi+tu} = \per\Omega_\phi + t \int_{\Sph^2} \left( u \sqrt{\phi^2 + (\nabla\phi)^2} + \phi \frac{\phi u + \nabla \phi\cdot\nabla u}{\sqrt{\phi^2 + (\nabla\phi)^2}} \right)d\omega +o(t) \,.
$$
Finally, for the interaction term we have
$$
\int_0^{\phi(\omega')+t u(\omega')} \frac{r'^2\,dr'}{|r\omega-r'\omega'|} = \int_0^{\phi(\omega')} \frac{r'^2\,dr'}{|r\omega-r'\omega'|} +  t u(\omega') \frac{\phi(\omega')^2}{|r\omega - \phi(\omega')\omega'|} + o(t) \,,
$$
so
\begin{align*}
& \int_0^{\phi(\omega)+t u(\omega)} \int_0^{\phi(\omega')+t u(\omega')} \frac{r'^2\,dr'\,r^2\,dr}{|r\omega-r'\omega'|} = \int_0^{\phi(\omega)} \int_0^{\phi(\omega')} \frac{r'^2\,dr'\,r^2\,dr}{|r\omega-r'\omega'|} \\
& \qquad + t u(\omega) \phi(\omega)^2 \int_0^{\phi(\omega')} \frac{r'^2\,dr'}{|\phi(\omega) \omega-r'\omega'|}
+  t u(\omega') \phi(\omega')^2 \int_0^{\phi(\omega)} \frac{r^2\,dr}{|r\omega - \phi(\omega')\omega'|} + o(t) \,.
\end{align*}
Thus, using \eqref{eq:coulombstarshaped},
\begin{align*}
D[\Omega_{\phi+ tu}] & = D[\Omega_\phi] + t \int_{\Sph^2}u(\omega) \phi(\omega)^2 \int_{\Sph^2} \int_0^{\phi(\omega')} \frac{r'^2\,dr'}{|\phi(\omega) \omega-r'\omega'|}\,d\omega'\,d\omega + o(t) \\
& = D[\Omega_\phi] + t \int_{\Sph^2}u(\omega) \phi(\omega)^2 \int_{\Omega_\phi} \frac{dx}{|\phi(\omega) \omega-x|} \,d\omega + o(t) \,.
\end{align*}

Putting everything together, we find
\begin{align*}
\mathcal I[ |\Omega_{\phi}|^{1/3} |\Omega_{\phi+tu}|^{-1/3} \Omega_{\phi+tu}] & = \mathcal I[\Omega_\phi] + t \left( \int_{\Sph^2} \left( u \sqrt{\phi^2 + (\nabla\phi)^2} + \phi \frac{\phi u + \nabla \phi\cdot\nabla u}{\sqrt{\phi^2 + (\nabla\phi)^2}} \right)d\omega \right. \\
& \quad + \int_{\Sph^2}u(\omega) \phi(\omega)^2 \int_{\Omega_\phi} \frac{dx}{|\phi(\omega) \omega-x|} \,d\omega \\
& \quad \left. - \left( \frac23 \per\Omega_\phi + \frac53 D[\Omega_\phi] \right)  |\Omega_\phi|^{-1} \int_{\Sph^2} \phi^2 u \,d\omega \right) + o(t) \,.
\end{align*}
Thus, $\phi$ is a volume-constrained critical point of $\mathcal I$ if and only if
\begin{align*}
& \int_{\Sph^2} \left( u \sqrt{\phi^2 + (\nabla\phi)^2} + \phi \frac{\phi u + \nabla \phi\cdot\nabla u}{\sqrt{\phi^2 + (\nabla\phi)^2}} \right)d\omega + \int_{\Sph^2}u(\omega) \phi(\omega)^2 \int_{\Omega_\phi} \frac{dx}{|\phi(\omega) \omega-x|} \,d\omega \\
& \qquad - \left( \frac23 \per\Omega_\phi + \frac53 D[\Omega_\phi] \right)  |\Omega_\phi|^{-1} \int_{\Sph^2} \phi^2 u \,d\omega =0
\end{align*}
for every Lipschitz $u:\Sph^2\to\R$, that is, if and only if
\begin{equation*}
-\nabla\cdot \frac{\phi \nabla\phi}{\sqrt{\phi^2 + (\nabla\phi)^2}} + \frac{\phi^2}{\sqrt{\phi^2+(\nabla\phi)^2}} + \sqrt{\phi^2+(\nabla\phi)^2} + \phi^2 \int_{\Omega_\phi} \frac{dy}{|\phi(\omega)\omega - y|}  = \mu \phi^2
\quad\text{on}\ \Sph^2
\end{equation*}
with $\mu$ from \eqref{eq:muvirial}. The latter equation is easily seen to be equivalent to \eqref{eq:nonlinear2}.
\end{proof}

The following result, although not necessary for the proof of our main results, clarifies the role of the parameter \eqref{eq:muvirial}. A similar statement with a different proof appears in the proof of \cite[Lemma 2]{Ju1}, see also \cite[Section 3]{Sw}.

\begin{lemma}
If \eqref{eq:nonlinear2} holds with some $\mu\in\R$, then $\mu$ is necessarily given by \eqref{eq:muvirial}.
\end{lemma}

\begin{proof}
The lemma follows by multiplying \eqref{eq:nonlinear2} by $\phi^3$ and integrating over $\Sph^2$ using
\begin{equation}
\label{eq:identity}
\int_{\Sph^2} \phi(\omega)^3 \int_{\Omega_\phi} \frac{dy}{|\phi(\omega)\omega - y|}\,d\omega = 5\, D[\Omega_\phi] \,.
\end{equation}
Let us prove the latter formula. We set
$$
V_\Omega(x) = \int_\Omega \frac{dy}{|x-y|} \,.
$$
Using the formulas for $d\sigma$ and $\nu$ from Subsection \ref{sec:prelim} we write the left side of \eqref{eq:identity} as
$$
\int_{\Sph^2} \phi(\omega)^3 \int_{\Omega_\phi} \frac{dy}{|\phi(\omega)\omega - y|}\,d\omega
= \int_{\partial\Omega_\phi}  \nu_x \cdot x \ V_{\Omega_\phi}(x)\,d\sigma(x) \,.
$$
We now prove that for any (sufficiently regular, but not necessarily star-shaped) set $\Omega\subset\R^3$
$$
\int_{\partial\Omega}  \nu_x \cdot x \ V_{\Omega}(x)\,d\sigma(x) = \frac{5}{2} \int_{\Omega} V_\Omega(x)\,dx \,.
$$
Indeed, by the divergence theorem we have
$$
\int_{\partial\Omega}  \nu_x \cdot x \ V_{\Omega}(x)\,d\sigma(x) = \int_\Omega \nabla\cdot (x V_\Omega(x))\,dx = 3 \int_\Omega V_\Omega(x)\,dx + \int_\Omega x\cdot\nabla V_\Omega(x)\,dx \,.
$$
Thus, the claim will follow provided we can show that
$$
\int_\Omega x\cdot\nabla V_\Omega(x)\,dx = - \frac12 \int_\Omega V_\Omega(x) \,dx \,.
$$
To prove this, we write
\begin{align*}
\int_\Omega x\cdot\nabla V_\Omega(x)\,dx & = - \iint_{\Omega\times\Omega} \frac{x\cdot(x-y)}{|x-y|^3} \,dx\,dy \\
& = - \iint_{\Omega\times\Omega} \frac{dx\,dy}{|x-y|} - \iint_{\Omega\times\Omega} \frac{y\cdot(x-y)}{|x-y|^3} \,dx\,dy \,.
\end{align*}
Renaming $x$ and $y$ we find
$$
\iint_{\Omega\times\Omega} \frac{y\cdot(x-y)}{|x-y|^3} \,dx\,dy = -  \iint_{\Omega\times\Omega} \frac{x\cdot(x-y)}{|x-y|^3} \,dx\,dy = \int_\Omega x\cdot\nabla V_\Omega(x)\,dx
$$
and inserting this into the previous identity, we obtain the claim.
\end{proof}


\section{Existence of a bifurcation}\label{sec:bifurc}

\subsection{Smoothness}
Our goal in this subsection is to prove Proposition \ref{smooth}, namely the smoothness of the map $\Phi:\mathcal O\to C^{0,\alpha}(\Sph^2)$. We will deduce this from bounds in \cite{CaFaWe2}, which deal with a much more singular situation. The observation that these bounds are also useful for more regular interaction kernels is from \cite{Fa}.

\begin{proof}[Proof of Proposition \ref{smooth}]
We split $F(\phi)=F_P(\phi) + F_C(\phi)$ with
$$
F_P(\phi) = -\nabla\cdot \frac{\nabla\phi}{\phi\sqrt{\phi^2+(\nabla\phi)^2}} + \frac{3}{\sqrt{\phi^2+(\nabla\phi)^2}} - \frac{\sqrt{\phi^2+(\nabla\phi)^2}}{\phi^2} 
$$
and
$$
F_C(\phi) = \int_{\Omega_\phi} \frac{dy}{|\phi(\omega)\omega-y|} \,.
$$
Clearly $F_P$ is $C^\infty$ as a map from $\{\phi\in C^{2,\alpha}(\Sph^2):\ \inf_{\Sph^2}\phi>0\}$ to $C^{0,\alpha}(\Sph^2)$. We now show that $F_C$ is $C^\infty$ as a map from $\{\phi\in C^{1,\alpha}(\Sph^2):\ \inf_{\Sph^2}\phi>0\}$ to $C^{0,\alpha}(\Sph^2)$, which will prove the claimed smoothness.

Using $\omega\cdot\nabla\phi(\omega)=0$ for every $\omega\in\Sph^2$ we rewrite the formula from Lemma \ref{potsurf} as
\begin{align}\label{eq:nonlocalrewrite}
- 2 F_C(\phi)(\omega) & = \phi(\omega) \int_{\Sph^2} \frac{\phi(\omega)-\phi(\omega')-(\omega-\omega')\cdot\nabla\phi(\omega')}{|\omega-\omega'|} K(\phi,\omega,\omega')\phi(\omega')\,d\omega' \notag \\
& \quad - \int_{\Sph^2} \frac{(\phi(\omega)-\phi(\omega'))^2}{|\omega-\omega'|} K(\phi,\omega,\omega')\phi(\omega')\,d\omega' \notag \\
& \quad - \frac{\phi(\omega)}{2} \int_{\Sph^2} |\omega-\omega'| K(\phi,\omega,\omega')\phi(\omega')^2\,d\omega 
\end{align}
with
$$
K(\phi,\omega,\omega') = \frac{|\omega-\omega'|}{|\phi(\omega)\omega-\phi(\omega')\omega'|} = \left( \frac{(\phi(\omega)-\phi(\omega'))^2}{|\omega-\omega'|^2} + \phi(\omega)\phi(\omega') \right)^{-1/2}.
$$
The right side of \eqref{eq:nonlocalrewrite} coincides with \cite[(4.17)]{CaFaWe2}, except for the fact that both in the factor $|\omega-\omega'|^{-1}$ and in the definition of $K(\phi,\omega,\omega')$ the exponent $N+\alpha$ in \cite[(4.17)]{CaFaWe2} is replaced by the exponent 1. Since $|\omega-\omega'|^{-1}$ is locally integrable, this both simplifies the proof and strengthens the result. Indeed, in the bound \cite[(4.47)]{CaFaWe2} a loss of $\alpha$ derivatives occurs which, we claim, does not happen in our situation. Once this is shown, the smoothness from $\{\phi\in C^{1,\alpha}(\Sph^2):\ \inf_{\Sph^2}\phi>0\}$ to $C^{0,\alpha}(\Sph^2)$ is shown by following the proof of \cite[Theorem 4.11]{CaFaWe2} line by line.

Thus, we only need to argue that if in the definition of $\mathcal F_1$ in \cite[Lemma 4.9]{CaFaWe2} the exponent $N+\alpha$ is replaced by $1$, then $C^{\beta-\alpha}$ on the left side of \cite[(4.47)]{CaFaWe2} can be replaced by $C^\beta$. We first note that the replacement of the exponents does not change the bounds on $K(\phi,\omega,\omega')$ and its derivatives in \cite[Lemma 4.8]{CaFaWe2}. Moreover, we can coarsen the bound \cite[(4.51)]{CaFaWe2} by estimating the minimum there by a constant times $|\theta_1-\theta_2|^\beta$ uniformly in $\sigma$. Using this bound we obtain \cite[(4.52)]{CaFaWe2} with the last factor on the right side replaced by $|\theta_1-\theta_2|^\beta$, which is already the claimed bound. This concludes the sketch of the proof.
\end{proof}


\subsection{The linearization}
Our goal in this subsection is to prove Propositions \ref{linearization} and~\ref{asscr}.

\begin{lemma}\label{linearizationlemma}
For $R>0$ and $u\in C^2(\Sph^2)$ one has pointwise on $\Sph^2$, as $t\to 0$,
$$
t^{-1} \left( F(R+tu) - F(R)\right) \to  R^{-2}( -\Delta u - 2u) + R \left( \int_{\Sph^2} \frac{u(\omega')}{|\omega-\omega'|}\,d\omega' - \frac{4\pi}{3} u \right).
$$
\end{lemma}

We omit the proof of this lemma, since we will compute a more precise expansion in Lemmas \ref{expper2} and \ref{expcou2} below. We are now in position to give the

\begin{proof}[Proof of Proposition \ref{linearization}]
Since we have already shown that $\Phi$ is Fr\'echet differentiable, we know that $D_\chi\Phi(R,0)[u]$ coincides with the pointwise limit of $t^{-1} R^{2}(F(R+tu)-F(R))$ as $t\to 0$. Thus, Lemma \ref{linearizationlemma} yields the claimed formula. From this formula it is clear that $D_\chi\Phi(R,0)$ commutes with rotations and therefore is diagonal in the basis of spherical harmonics. Moreover, it is well-known that the eigenvalue of $-\Delta$ on the space of spherical harmonics of degree $\ell\in\N_0$ is $\ell(\ell+1)$. Moreover, by the Funk--Hecke formula, the eigenvalue of the integral operator with integral kernel $|\omega-\omega'|^{-1} = (2(1-\omega\cdot\omega'))^{-1/2}$ on that space is equal to
$$
2\pi \int_{-1}^1 \frac{P_\ell(t)}{\sqrt{2(1-t)}} \,dt \,,
$$
where $P_\ell$ is the $\ell$-th Legendre polynomial. We now use the fact that for $|a|<1$ and $t\in[-1,1]$,
$$
\frac{1}{\sqrt{1-2at + a^2}} = \sum_{\ell=0}^\infty P_\ell(t) a^\ell \,.
$$
We apply this with $a=1-\epsilon$ and, using
$$
\int_{-1}^1 P_\ell(t) P_{\ell'}(t)\,dt = \frac{2}{2\ell+1} \delta_{\ell,\ell'} \,,
$$
obtain
\begin{align*}
\int_{-1}^1 \frac{P_\ell(t)}{\sqrt{2(1-t + \epsilon^2/(2(1-\epsilon)))}} \,dt & = 
\sqrt{1-\epsilon} \int_{-1}^1 \frac{P_\ell(t)}{\sqrt{(1-2(1-\epsilon) t +(1-\epsilon)^2}} \,dt \\
& = \frac{2}{2\ell+1} (1-\epsilon)^{\ell+1/2} \,.
\end{align*}
Using the fact that $P_\ell(t)$ is bounded we obtain by dominated convergence
$$
\int_{-1}^1 \frac{P_\ell(t)}{\sqrt{2(1-t)}} \,dt = \frac{2}{2\ell+1} \,.
$$
This proves the claimed formula for the eigenvalue.
\end{proof}


\begin{proof}[Proof of Proposition \ref{asscr}]
According to the formula for the eigenvalues from Proposition~\ref{linearization}, the kernel of $D_\chi\Phi(R,0)$ is equal to the space of spherical harmonics of degree one and two. The intersection of this space with $X$ is spanned by $P$.

Let us compute the range of $L_{R_*}$. The inclusion $\subset$ in the proposition is easy. To prove the opposite inclusion, let $\chi\in Y$ with $\int P\chi\,d\omega =0$. Then, in particular, $\chi\in L^2(\Sph^2)$. Because of the explicit form of the spectrum we see that the operator $L_{R_*}$ maps $\{ u\in H^2(\Sph^2):\ u \ \text{depends only on}\ |\omega_3|\}$ onto $\{ f\in L^2(\Sph^2):\ f \ \text{depends only on}\ |\omega_3|\}$. Thus, there is a $u\in H^2(\Sph^2)$ depending only on $|\omega_3|$ such that $L_{R_*} u = \chi$. We will show that $u\in C^{2,\alpha}(\Sph^2)$. It is easy to see that
$$
\int_{\Sph^2} \frac{g(\omega')}{|\omega-\omega'|}\,d\omega'
$$
belongs to $C^{0,\alpha}(\Sph^2)$ for any $g\in H^2(\Sph^2)$. (In fact, this is true for much less regular $g$.) By Morrey's embedding theorem, the function $u\in H^2(\Sph^2)$ belongs to $C^{0,\alpha}(\Sph^2)$ (no matter how close $\alpha$ is to $1$). Thus,
$$
-\Delta u = 2 u + R_*^3 \left(  \int_{\Sph^2} \frac{u(\omega')}{|\omega-\omega'|}\,d\omega' - \frac{4\pi}{3} u \right) +\chi \in C^{0,\alpha}(\Sph^2) \,.
$$
By elliptic regularity, $u\in C^{2,\alpha}(\Sph^2)$ and therefore $\chi =L_{R_*} u\in\ran L_{R_*}$, as claimed.

Finally, using the explicit form of the eigenvalues of $L_R$ from Proposition \ref{linearization},
$$
L_R P = \left( 4 - \frac{8\pi}{15} R^3 \right) P \,.
$$
and, therefore,
$$
\frac{d}{dR} L_R P = - \frac{8\pi}{5} R^2 P \,.
$$
From the characterization of $\ran L_{R_*}$ we obtain the last assertion.
\end{proof}


\section{The second derivative}\label{sec:second}

Our goal in this section is to prove Proposition \ref{secondderivative}. To do so, we split $F(\phi)=F_P(\phi) + F_C(\phi)$ as in the proof of Proposition \ref{smooth}. We expand both terms to second order around a constant.

\begin{lemma}\label{expper2}
For $R>0$ and $u\in C^2(\Sph^2)$ one has pointwise on $\Sph^2$, as $t\to 0$,
$$
F_P(R+tu) = \frac{2}{R} + \frac{t}{R^2} \left( -\Delta u - 2u\right) + \frac{2 t^2}{R^3} \left( u\Delta u + u^2 \right) + o(t^2) \,.
$$ 
\end{lemma}

\begin{proof}
We set $\phi=R+tu$ and compute
$$
\sqrt{\phi^2+(\nabla\phi)^2} = R + t u + \frac{t^2}{2R} (\nabla u)^2 + o(t^2) \,.
$$
Therefore,
\begin{align*}
\frac{3}{\sqrt{\phi^2+(\nabla\phi)^2}}
& = \frac{3}{R} - \frac{3tu}{R^2} + \frac{3t^2}{R^3} \left( -\frac{1}{2} (\nabla u)^2 + u^2 \right) + o(t^2)
\end{align*}
and
\begin{align*}
\frac{\sqrt{\phi^2+(\nabla\phi)^2}}{\phi^2} & = R^{-2} \left( R - t u + \frac{t^2}R \left( \frac12(\nabla u)^2 + u^2 \right) + o(t^2) \right).
\end{align*}
Finally,
\begin{align*}
\frac{1}{\phi\sqrt{\phi^2+(\nabla\phi)^2}}
&= R^{-2} \left( 1 - \frac{2tu}{R} + o(t) \right),
\end{align*}
so
$$
\nabla\cdot \frac{\nabla\phi}{\phi\sqrt{\phi^2+(\nabla\phi)^2}} = \frac{t}{R^2} \nabla\cdot \left( \left( 1 - \frac{2tu}{R} + o(t) \right)\nabla u \right) = \frac{t}{R^2}\Delta u - \frac{2t^2}{R^3} \nabla\cdot \left( u\nabla u \right) + o(t^2) \,.
$$
Collecting all the terms we find
$$
F_P(\phi) = \frac{2}{R} + \frac{t}{R^2} \left( -\Delta u - 2u\right) + \frac{2 t^2}{R^3} \left( \nabla\cdot(u\nabla u) - (\nabla u)^2 + u^2 \right) + o(t^2) \,.
$$ 
Using $\nabla\cdot(u\nabla u) = (\nabla u)^2 + u\Delta u$, we obtain the assertion.
\end{proof}

\begin{lemma}\label{expcou2}
For $R>0$ and $u\in C^{0,\alpha}(\Sph^2)$ for some $0<\alpha<1$ one has pointwise on $\Sph^2$, as $t\to 0$,
\begin{align*}
F_C(R+tu) & = \frac{4\pi}{3} R^2 + tR \left( \int_{\Sph^2} \frac{u(\omega')}{|\omega-\omega'|}\,d\omega' - \frac{4\pi}{3} u(\omega) \right) \\
& \quad + t^2 \left( \frac{\pi}{3} u(\omega)^2 - \frac12 u(\omega) \int_{\Sph^2} \frac{u(\omega')}{|\omega-\omega'|}\,d\omega' + \frac{3}{4} \int_{\Sph^2} \frac{u(\omega')^2}{|\omega-\omega'|}\,d\omega' \right) + o(t^2) \,.
\end{align*}
\end{lemma}

\begin{proof}
Again we write $\phi=R+t u$. Our starting point is the formula
\begin{align}\label{eq:excou2proof0}
\int_{\Omega_\phi}\frac{dy}{|\phi(\omega)\omega-y|} & = \int_{\Sph^2} \int_0^{\phi(\omega')} \frac{r^2 \,dr \,d\omega'}{|\phi(\omega)\omega -r\omega'|} = \phi(\omega)^2 \int_{\Sph^2} \int_0^{\phi(\omega')/\phi(\omega)} \frac{s^2\,ds\,d\omega'}{|\omega - s\omega'|} \notag \\
& = \phi(\omega)^2 \int_{\Sph^2} \int_0^{1+t g(\omega,\omega')} f(s,\omega,\omega')\,ds\,d\omega' \,,
\end{align}
where
$$
g(\omega,\omega') = t^{-1}\left(\frac{\phi(\omega')}{\phi(\omega)}-1\right) =  \frac{u(\omega')-u(\omega)}{R+tu(\omega)}
\qquad\text{and}\qquad
f(s,\omega,\omega') = \frac{s^2}{|\omega-s\omega'|} \,.
$$
(We do not reflect the $t$-dependence of $g$ in the notation.) Dropping for the moment the $\omega$ and $\omega'$-dependence from the notation as well, we write
\begin{align*}
\int_0^{1+t g} f(s)\,ds & = \int_0^1 f(s)\,ds + t g \int_0^1 f(1+tg\sigma)\,d\sigma \\
& = \int_0^1 f(s)\,ds + t g f(1) + t^2 g^2 \int_0^1 \sigma \int_0^1 \partial_s f(1+tg\sigma\tau)\,d\tau\,d\sigma \,.
\end{align*}
We compute
$$
\partial_s f(s) = \frac{3}{2}\frac{s}{|\omega-s\omega'|} + \frac{1}{2}\frac{s(1-s^2)}{|\omega-s\omega'|^3}
$$
and bound, using $|\omega-s\omega'|^2 = (1-s)^2 + s|\omega-\omega'|^2$,
$$
|\partial_s f(s)| \leq \frac{3}{2}\frac{\sqrt s}{|\omega-\omega'|} + \frac{1}{2} \frac{1+s}{|\omega-\omega|^2} \,.
$$
This implies that for every $\omega,\omega'\in\Sph^2$, as $t\to 0$,
$$
g(\omega,\omega')^2 \int_0^1 \sigma \int_0^1 \partial_s f(1+tg(\omega,\omega')\sigma\tau,\omega,\omega')\,d\tau\,d\sigma
\to \frac34 \frac{(u(\omega')-u(\omega))^2}{R^2 |\omega-\omega'|} \,.
$$
Moreover, if $t\leq R/(2\|u\|_\infty)$, then
$$
\left| g(\omega,\omega')^2 \int_0^1 \sigma \int_0^1 \partial_s f(1+tg(\omega,\omega')\sigma\tau,\omega,\omega')\,d\tau\,d\sigma \right|
\leq C \frac{(u(\omega')-u(\omega))^2}{R^2} \frac{1}{|\omega-\omega'|^2}
$$
with a universal constant $C<\infty$. Since $u\in C^{0,\alpha}(\Sph^2)$ for some $0<\alpha<1$, the right side is integrable in $\omega'$ and therefore dominated convergence implies that
\begin{align}
\label{eq:excou2proof}
\int_{\Sph^2} \int_0^{1+t g(\omega,\omega')} f(s,\omega,\omega')\,ds\,d\omega' & = \int_{\Sph^2} \int_0^1 f(s,\omega,\omega')\,ds\,d\omega' + t \int_{\Sph^2} g(\omega,\omega') f(1,\omega,\omega')\,d\omega \notag \\
& \quad + \frac{t^2}{R^2} \frac{3}{4} \int_{\Sph^2} \frac{(u(\omega')-u(\omega))^2}{|\omega-\omega'|}\,d\omega' +o(t^2) \,.
\end{align}
The leading term in \eqref{eq:excou2proof} is
$$
\int_{\Sph^2} \int_0^{1} f(s,\omega,\omega') \,ds\,d\omega' = \int_B \frac{dy}{|\omega-y|} = \frac{4\pi}{3} \,.
$$
Using
$$
g(\omega,\omega') = \frac{u(\omega')-u(\omega)}{R} - t \frac{u(\omega)(u(\omega')-u(\omega))}{R^2} + o(t)
$$
and
\begin{equation}\label{eq:newton}
\int_{\Sph^2} \frac{d\omega'}{|\omega-\omega'|} = 4\pi \,.
\end{equation}
we obtain for the second term on the right side of \eqref{eq:excou2proof} that
\begin{align*}
t \int_{\Sph^2} g(\omega,\omega') f(1,\omega,\omega')\,d\omega
& = \frac{t}{R} \left( \int_{\Sph^2} \frac{u(\omega')}{|\omega-\omega'|}\,d\omega' - 4\pi u(\omega) \right) \\
& \quad + \frac{t^2}{R^2} \left( - u(\omega) \int_{\Sph^2} \frac{u(\omega')}{|\omega-\omega'|}\,d\omega' + 4\pi u(\omega)^2 \right) + o(t^2) \,.
\end{align*}
Finally, using again \eqref{eq:newton} we rewrite the last term in \eqref{eq:excou2proof} as
$$
\int_{\Sph^2} \frac{(u(\omega')-u(\omega))^2}{|\omega-\omega'|}\,d\omega'
= \int_{\Sph^2} \frac{u(\omega')^2}{|\omega-\omega'|}\,d\omega' - 2 u(\omega) \int_{\Sph^2} \frac{u(\omega')}{|\omega-\omega'|}\,d\omega' + 4\pi u(\omega)^2 \,.
$$
Inserting the expansion of \eqref{eq:excou2proof} into \eqref{eq:excou2proof0} we easily obtain the formula in the lemma.
\end{proof}

We are now in position to give the

\begin{proof}[Proof of Proposition \ref{secondderivative}]
We introduce
$$
\mathcal Q_R[u](\omega) = \frac{2}{R^3} \left( u\Delta u + u^2\right) +  \frac{\pi}{3} u(\omega)^2 - \frac12 u(\omega) \int_{\Sph^2} \frac{u(\omega')}{|\omega-\omega'|}\,d\omega' + \frac{3}{4} \int_{\Sph^2} \frac{u(\omega')^2}{|\omega-\omega'|}\,d\omega' \,.
$$
Then Lemmas \ref{expper2} and \ref{expcou2} imply that
$$
t^{-2} \left( \Phi(R,tu) - t D_\chi\Phi(R,0)[u] \right) = t^{-2} R^2 \left( F(R+tu) - F(R) - t R^{-2} L_R u \right) \to R^2 \mathcal Q_R[u]
$$
pointwise on $\Sph^2$. Since we have already shown that $\Phi$ is twice Fr\'echet differentiable, we conclude that
$$
\frac12 D_{\chi\chi}^2 \Phi(R,0)[u,u] = R^2 \mathcal Q_R[u] \,.
$$
Thus, Proposition \ref{secondderivative} will follow if we can show that
\begin{equation}
\label{eq:secondderproof}
\mathcal Q_{R_*}[P](\omega) = -\frac{8\pi}{5} \left( \frac{12}{35} P_4(\omega_3) + \frac{1}{7} P_2(\omega_3) - \frac{1}{5} \right) \,.
\end{equation}
Since $P$ is a spherical harmonic of degree two, we have, as in the proof of Proposition~\ref{linearization},
\begin{equation*}
-\Delta P = 6 P
\qquad\text{and}\qquad
\int_{\Sph^2} \frac{P(\omega')}{|\omega-\omega'|}\,d\omega' = \frac{4\pi}{5} P(\omega) \,,
\end{equation*}
so
$$
\frac{2}{R_*^3} (P\Delta P + P^2) + \frac{\pi}{3} P^2 -\frac12 P \int_{\Sph^2} \frac{P(\omega')}{|\omega-\omega'|}\,d\omega' = -\frac{7}{5} \pi P^2 \,.
$$
We now use the explicit form of the Legendre polynomials to write
$$
P_2(t)^2 = \frac{9 t^4 - 6 t^2 + 1}4 = \frac{18}{35} P_4(t) + \frac{2}{7} P_2(t) + \frac{1}{5} \,,
$$
so
\begin{align*}
\frac{2}{R_*^3} (P\Delta P + P^2) + \frac{\pi}{3} P^2 -\frac12 P \int_{\Sph^2} \frac{P(\omega')}{|\omega-\omega'|}\,d\omega' = -\frac{7}{5} \pi \left( \frac{18}{35} P_4(\omega_3) + \frac{2}{7} P_2(\omega_3) + \frac{1}{5} \right).
\end{align*}
Using the formula for $P_2^2$ again and as well as the formula for the eigenvalues of the operator with integral kernel $|\omega-\omega'|^{-1}$ from Proposition \ref{linearization}, we also find, recalling that $P_\ell(\omega_3)$ is a spherical harmonic of degree $\ell$,
\begin{align*}
\int_{\Sph^2} \frac{P(\omega)^2}{|\omega-\omega'|}\,d\omega' = \frac{18}{35} \ \frac{4\pi}{9} P_4(\omega_3) + \frac{2}{7} \ \frac{4\pi}{5} P_2(\omega_3) + \frac{1}{5} 4\pi \,.
\end{align*}
Multiplying this formula by $3/4$ and adding it to the previous formula, we obtain \eqref{eq:secondderproof}. This concludes the proof of the proposition.
\end{proof}



\section{Expansion of the energy}\label{sec:energy}

Our goal in this section is to prove Theorem \ref{energy} concerning the difference in energy between $\Omega_{R_s+\chi_s}$ and the ball of the same volume. As a preparation for the proof, in the following two lemmas we compute the perimeter and the Coulomb energy of almost spherical sets up to third order in the deviation from a constant.

\begin{lemma}\label{perimetersecondorder}
As $t\to 0$,
\begin{align}
\label{eq:perimeter}
\per \Omega_{R+t u} & = 4\pi R^2 + 2tR \int_{\Sph^2} u\,d\omega + t^2 \left( \frac12 \int_{\Sph^2} (\nabla u)^2\,d\omega + \int_{\Sph^2} u^2\,d\omega \right) + O(t^4) \,.
\end{align}
This expansion is uniform for $(R,u)$ from bounded sets in $(0,\infty)\times C^{0,1}(\Sph^2)$.
\end{lemma}

\begin{proof}
We expand pointwise
$$
(R+tu)\sqrt{(R+tu)^2+t^2(\nabla u)^2} = R^2 + 2tRu + \frac12 t^2(\nabla u)^2 + t^2 u^2 + \mathcal O(t^4) \,.
$$
The assertion follows by integration using \eqref{eq:perimeterstarshaped}.
\end{proof}

\begin{lemma}\label{coulombsecondorder}
As $t\to 0$,
\begin{align}
\label{eq:coulomb}
D[\Omega_{R+t u}] & = \frac{(4\pi)^2}{15} R^5 + \frac{4\pi}{3} R^4 t \int_{\Sph^2} u(\omega)\,d\omega \notag \\
& \quad + \frac12 t^2 R^3 \left( \frac{4\pi}{3} \int_{\Sph^2} u(\omega)^2\,d\omega + \iint_{\Sph^2\times\Sph^2} \frac{u(\omega)u(\omega')}{|\omega-\omega'|}\,d\omega\,d\omega' \right) \notag \\
& \quad + t^3 R^2 \left( - \frac{\pi}{3} \int_{\Sph^2} u(\omega)^3\,d\omega + \frac34 \iint_{\Sph^2\times\Sph^2} \frac{u(\omega)^2u(\omega')}{|\omega-\omega'|}\,d\omega\,d\omega' \right) + O(t^4) \,.
\end{align}
This expansion is uniform for $(R,u)$ from bounded sets in $(0,\infty)\times C^{0,1}(\Sph^2)$.
\end{lemma}

The proof uses some ideas from the proof of \cite[Theorem 2.1]{FiFuMaMiMo}, where a similar expansion up to order $t^2$ is obtained.

\begin{proof}
Using \eqref{eq:coulombstarshaped} we write
$$
D[\Omega_{R+t u}] = \frac12 \iint_{\Sph^2\times\Sph^2} \left( I_1(\omega,\omega') - I_2(\omega,\omega') \right) d\omega\,d\omega'
$$
with
\begin{align*}
I_1(\omega,\omega') & =
\frac12 \int_0^{R+t u(\omega)} \int_0^{R+t u(\omega)} F(r,r',\omega,\omega') \,dr\,dr' \\
& \quad + \frac12 \int_0^{R+t u(\omega')} \int_0^{R+t u(\omega')} F(r,r',\omega,\omega') \,dr\,dr'
\end{align*}
and
$$
I_2(\omega,\omega') = \frac12 \int_{R+t u(\omega)}^{R+t u(\omega')} \int_{R+t u(\omega)}^{R+t u(\omega')} F(r,r',\omega,\omega') \,dr\,dr' \,,
$$
where
$$
F(r,r',\omega,\omega') = \frac{r^2\, r'^2}{|r\omega-r'\omega'|} \,.
$$

We begin by discussing the term involving $I_1$. By scaling we have
$$
\int_0^a \int_0^a F(r,r',\omega,\omega') \,dr\,dr' = a^5 \int_0^1 \int_0^1 \frac{s^2\,ds\,s'^2\,ds'}{|s\omega-s'\omega|}
$$
and therefore
\begin{align*}
I_1(\omega,\omega') = \frac{(R+t u(\omega))^5 + (R+t u(\omega'))^5}{2} \int_0^1 \int_0^1 \frac{s^2\,ds\,s'^2\,ds'}{|s\omega-s'\omega|}
\end{align*}
and, by symmetry,
\begin{align*}
\frac12 \iint_{\Sph^2\times\Sph^2} I_1(\omega,\omega') \,d\omega\,d\omega' & = \int_{\Sph^2} (R+t u(\omega))^5 \frac12 \int_{\Sph^2} \int_0^1 \int_0^1 \frac{s^2\,ds\,s'^2\,ds'}{|s\omega-s'\omega|} d\omega' \,d\omega \notag \\ 
& = \frac{4\pi}{15}\int_{\Sph^2} (R+t u(\omega))^5 \,d\omega \notag \\
& = \frac{(4\pi)^2}{15} R^5 + \frac{4\pi}{3} R^4 t \int_{\Sph^2} u(\omega)\,d\omega + \frac{8\pi}{3} R^3 t^2 \int_{\Sph^2} u(\omega)^2\,d\omega \notag \\
& \quad + \frac{8\pi}{3} R^2 t^3 \int_{\Sph^2} u(\omega)^3\,d\omega + O(t^4) \,.
\end{align*}
Here we used the fact that
$$
\int_{\Sph^2} \int_0^1 \frac{s'^2\,ds'}{|s\omega-s'\omega|}\,d\omega'
 = \frac{1}{4\pi} \iint_{\Sph^2\times\Sph^2} \int_0^1 \frac{s'^2\,ds'}{|s\omega-s'\omega|}\,d\omega'\,d\omega
$$
(since the integral on the left is independent of $\omega$) and, consequently,
$$
\frac12 \int_{\Sph^2} \int_0^1 \int_0^1 \frac{s^2\,ds\,s'^2\,ds'}{|s\omega-s'\omega|} d\omega' = \frac{1}{4\pi} D[B] = \frac{4\pi}{15} \,.
$$

We now discuss the term involving $I_2$ and write
$$
I_2(\omega,\omega') = \frac{t^2}{2} \int_{u(\omega)}^{u(\omega')} \int_{u(\omega)}^{u(\omega')} F(R+t\rho,R+t\rho',\omega,\omega')\,d\rho\,d\rho' \,.
$$
One easily proves the pointwise bound
\begin{align*}
& \left| F(1+s\rho,1+s\rho',\omega,\omega') - \frac{1}{|\omega-\omega'|} - s \frac{3}{2}\, \frac{\rho+\rho'}{|\omega-\omega'|} \right| \\
&  = \left| F(1+s\rho,1+s\rho',\omega,\omega')-F(1,1,\omega,\omega')- s(\rho\partial_r F(1,1,\omega,\omega') + \rho'\partial_{r'}F(1,1,\omega,\omega') \right| \\
& \leq C \frac{s^2 (\rho^2+\rho'^2)}{|\omega-\omega'|^3}
\qquad\text{if}\ \max\{|s\rho|,|s\rho'|\}\leq 1/2 \,,
\end{align*}
which implies by scaling that
\begin{align*}
& \left| F(R+t\rho,R+t\rho',\omega,\omega')- \frac{R^3}{|\omega-\omega'|} - t \frac{3}{2}\, \frac{R^2(\rho+\rho')}{|\omega-\omega'|} \right| \leq C R \frac{t^2 (\rho^2+\rho'^2)}{|\omega-\omega'|^3}
\end{align*}
if $\max\{|t\rho|,|t\rho'|\}\leq R/2$. By integration with respect to $\rho$ and $\rho'$ we obtain
\begin{align*}
& \left| I_2(\omega,\omega') - \frac{t^2}{2}R^3 \frac{(u(\omega')-u(\omega))^2}{|\omega-\omega'|} - \frac{3 t^3}{4} \frac{R^2(u(\omega')-u(\omega))^2 (u(\omega')+u(\omega))}{|\omega-\omega'|} \right| \\
& \leq \frac{C}{3} \, t^4 \,\frac{R(u(\omega')-u(\omega')) (u(\omega')^3-u(\omega)^3)}{|\omega-\omega'|^3} \leq C L_u^2 \|u\|_\infty^2 \, t^4 \,\frac{R}{|\omega-\omega'|}
\end{align*}
if $|t| \|u\|_\infty \leq R/2$, where $L_u = \sup_{\omega,\omega'} \frac{|u(\omega)-u(\omega')|}{|\omega-\omega'|}$. By integration with respect to $\omega$ and $\omega'$ we obtain
\begin{align*}
\frac{1}{2} \iint_{\Sph^2\times\Sph^2} I_2(\omega,\omega')\,d\omega\,d\omega' & = \frac{1}{4} t^2 R^3 \iint_{\Sph^2\times\Sph^2} \frac{(u(\omega')-u(\omega))^2}{|\omega-\omega'|} \,d\omega\,d\omega' \\
& \quad + \frac{3}{8} t^3 R^2 \iint_{\Sph^2\times\Sph^2} \frac{(u(\omega')-u(\omega))^2 (u(\omega)+u(\omega'))}{|\omega-\omega'|} \,d\omega\,d\omega' \\
& \quad + O(t^4) \,.
\end{align*}
This expansion is uniform for $(R,u)$ from bounded sets in $(0,\infty)\times C^{0,1}(\Sph^2)$. We write the first term on the right side as
\begin{align*}
\frac12 \iint_{\Sph^2\times\Sph^2} \!\!\frac{(u(\omega')-u(\omega))^2}{|\omega-\omega'|} \,d\omega\,d\omega' & = \int_{\Sph^2} \! u(\omega)^2\! \int_{\Sph^2} \frac{d\omega'}{|\omega-\omega'|}\,d\omega - \iint_{\Sph^2\times\Sph^2}\!\! \frac{u(\omega)u(\omega')}{|\omega-\omega'|}\,d\omega\,d\omega' \\
& = 4\pi \int_{\Sph^2} \! u(\omega)^2 \,d\omega - \iint_{\Sph^2\times\Sph^2} \!\!\frac{u(\omega)u(\omega')}{|\omega-\omega'|}\,d\omega\,d\omega' \,.
\end{align*}
where we used \eqref{eq:newton}. Similarly,
\begin{align*}
& \frac12 \iint_{\Sph^2\times\Sph^2} \!\!\frac{(u(\omega')-u(\omega))^2(u(\omega)+u(\omega'))}{|\omega-\omega'|} \,d\omega\,d\omega' \\
& \quad = \int_{\Sph^2} \! u(\omega)^3\! \int_{\Sph^2} \frac{d\omega'}{|\omega-\omega'|}\,d\omega - \iint_{\Sph^2\times\Sph^2}\!\! \frac{u(\omega)^2 u(\omega')}{|\omega-\omega'|}\,d\omega\,d\omega' \\
& \quad = 4\pi \int_{\Sph^2} \! u(\omega)^3 \,d\omega - \iint_{\Sph^2\times\Sph^2} \!\!\frac{u(\omega)^2 u(\omega')}{|\omega-\omega'|}\,d\omega\,d\omega' \,.
\end{align*}
This shows that
\begin{align*}
\frac{1}{2} \iint_{\Sph^2\times\Sph^2} I_2(\omega,\omega')\,d\omega\,d\omega' & = \frac{1}{2} t^2 R^3 \left( 4\pi \int_{\Sph^2} u(\omega)^2 \,d\omega - \iint_{\Sph^2\times\Sph^2} \frac{u(\omega)u(\omega')}{|\omega-\omega'|}\,d\omega\,d\omega' \right) \\
& \quad + \frac{3}{4} t^3 R^2 \left( 4\pi \int_{\Sph^2} \! u(\omega)^3 \,d\omega - \iint_{\Sph^2\times\Sph^2} \!\!\frac{u(\omega)^2 u(\omega')}{|\omega-\omega'|}\,d\omega\,d\omega' \right) \\
& \quad + O(t^4) \,.
\end{align*}
This completes the proof of \eqref{eq:coulomb}.
\end{proof}

\begin{proof}[Proof of Theorem \ref{energy}]
The fact that the map $s\mapsto\chi_s$ from Theorem \ref{main} is $C^4$ implies that there are functions $Q,T\in C^{2,\alpha}(\Sph^2)$ such that $\chi_s = sP+s^2 Q+s^3 T +\mathcal O(s^4)$ in $C^{2,\alpha}$ as $s\to 0$. We have computed the function $Q$ explicitly in the proof of Theorem \ref{secondorder}, but for the proof of Theorem \ref{energy} the pure existence of this function, as well as that of $T$, suffices. On the other hand, we will use the explicit form \eqref{eq:secondorderr} of the coefficient of $s$ in the expansion of $R_s$.

Using formula \eqref{eq:volumestarshaped} we obtain
\begin{align}\label{eq:expvol}
|\Omega_{R_s+\chi_s}| & = \frac{4\pi}{3} R_s^3 + R_s^2 \int_{\Sph^2} \chi_s\,d\omega + R_s \int_{\Sph^2} \chi_s^2\,d\omega + \frac13 \int_{\Sph^2} \chi_s^3\,d\omega \notag \\
& = \frac{4\pi}{3} R_s^3 + s^2 \left( R_s^2 \int_{\Sph^2} Q\,d\omega + R_s \int_{\Sph^2} P^2\,d\omega \right) \notag \\
& \quad + s^3 \left( R_s^2 \int_{\Sph^2} T\,d\omega + \frac13 \int_{\Sph^2} P^3\,d\omega \right) + \mathcal O(s^4) \,,
\end{align}
where in the last equality we used the facts that 
\begin{equation}
\label{eq:orthogonalityexp}
\int_{\Sph^2} P\,d\omega = 0
\qquad\text{and}\qquad
\int_{\Sph^2} PQ\,d\omega = 0 \,.
\end{equation}
The second relation follows from \eqref{eq:orthogonality}. Thus,
\begin{align*}
\rho_s = \left( \frac{3}{4\pi} |\Omega_{R_s+\chi_s}| \right)^{1/3} & = R_s + s^2 \left(  \frac{1}{4\pi} \int_{\Sph^2} Q\,d\omega + R_s^{-1} \frac{1}{4\pi} \int_{\Sph^2} P^2\,d\omega \right) \\
& \quad + s^3 \left( \frac{1}{4\pi} \int_{\Sph^2} T\,d\omega + \frac{1}{12\pi} R_s^{-2} \int_{\Sph^2} P^3\,d\omega \right) + \mathcal O(s^4) \,.  
\end{align*}
For later purposes we record that this implies $\rho_s=R_s+\mathcal O(s^2)$ and therefore
\begin{align}\label{eq:Rexp}
R_s & = \rho_s - s^2 \left(  \frac{1}{4\pi} \int_{\Sph^2} Q\,d\omega + \rho_s^{-1} \frac{1}{4\pi} \int_{\Sph^2} P^2\,d\omega \right)  \notag \\
& \quad - s^3 \left( \frac{1}{4\pi} \int_{\Sph^2} T\,d\omega + \frac{1}{12\pi} \rho_s^{-2} \int_{\Sph^2} P^3\,d\omega \right) + \mathcal O(s^4) \,.
\end{align}

From Lemma \ref{perimetersecondorder} we obtain
\begin{align}\label{eq:expper}
\per \Omega_{R_s+\chi_s} & = 4\pi R_s^2 + 2 R_s \int_{\Sph^2}\chi_s\,d\omega + \frac{1}{2} \int_{\Sph^2} (\nabla\chi_s)^2\,d\omega + \int_{\Sph^2} \chi_s^2\,d\omega + \mathcal O(s^4) \notag \\
& = 4\pi R_s^2 + s^2 \left( 2R_s \int_{\Sph^2} Q\,d\omega + \frac12 \int_{\Sph^2} (\nabla P)^2\,d\omega  + \int_{\Sph^2} P^2\,d\omega \right) \notag \\
& \quad + s^3 2R_s \int_{\Sph^2} T\,d\omega + \mathcal O(s^4) \,,
\end{align}
where we used \eqref{eq:orthogonalityexp} as well as
$$
\int_{\Sph^2} \nabla P\cdot\nabla Q\,d\omega = 0 \,.
$$
This follows from the second relation in \eqref{eq:orthogonalityexp} since $-\Delta P$ is proportional to $P$. Similarly, from Lemma \ref{coulombsecondorder} we obtain
\begin{align}\label{eq:expcou}
D[\Omega_{R_s+\chi_s}] & = \frac{(4\pi)^2}{15} R_s^5 + \frac{4\pi}{3} R_s^4 \int_{\Sph^2}\chi_s(\omega)\,d\omega \notag \\
& \qquad + \frac12 R_s^3 \left( \frac{4\pi}{3}\int_{\Sph^2} \chi_s(\omega)^2\,d\omega + \iint_{\Sph^2\times\Sph^2} \frac{\chi_s(\omega)\,\chi_s(\omega')}{|\omega-\omega'|}\,d\omega\,d\omega' \right) \notag \\
& \qquad + R_s^2 \left( - \frac{\pi}{3} \int_{\Sph^2} \chi_s(\omega)^3\,d\omega + \frac{3}{4} \iint_{\Sph^2\times\Sph^2} \frac{\chi_s(\omega)^2\chi_s(\omega')}{|\omega-\omega'|}\,d\omega\,d\omega' \right) + \mathcal O(s^4) \notag \\
& = \frac{(4\pi)^2}{15} R_s^5 
+  
s^2 \left( \frac{4\pi}{3} R_s^4 \int_{\Sph^2} Q(\omega)\,d\omega
+ \frac{2\pi}{3} R_s^3 \int_{\Sph^2} P(\omega)^2\,d\omega \right. \notag \\
& \qquad\qquad\qquad\qquad \left. + \frac12 R_s^3 \iint_{\Sph^2\times\Sph^2} \frac{P(\omega)\,P(\omega')}{|\omega-\omega'|} \,d\omega\,d\omega' \right) \notag \\
& \qquad 
+
s^3 \left( \frac{4\pi}{3} R_s^4 \int_{\Sph^2} T(\omega)\,d\omega - \frac{\pi}{3} R_s^2 \int_{\Sph^2} P(\omega)^3\,d\omega \right. \notag \\
& \qquad\qquad\qquad\qquad \left. + \frac{3}{4} R_s^2 \iint_{\Sph^2\times\Sph^2} \frac{P(\omega)^2\, P(\omega')}{|\omega-\omega'|}\,d\omega\,d\omega' \right) + \mathcal O(s^4) 
\,. 
\end{align}
Here we used \eqref{eq:orthogonalityexp} as well as
$$
\iint_{\Sph^2\times\Sph^2} \frac{P(\omega)\,Q(\omega')}{|\omega-\omega'|} \,d\omega\,d\omega' = 0 \,.
$$
This follows from the second relation in \eqref{eq:orthogonalityexp} since $\int P(\omega) |\omega-\omega'|^{-1}\,d\omega'$ is proportional to $P(\omega)$ by the Funk--Hecke formula as in the proof of Proposition \ref{linearization}.

Inserting \eqref{eq:Rexp} into \eqref{eq:expper} and \eqref{eq:expcou} we obtain
\begin{align*}
\per \Omega_{R_s+\chi_s} = 4\pi \rho_s^2
+ s^2 \left( \frac12 \int_{\Sph^2} (\nabla P)^2\,d\omega  - \int_{\Sph^2} P^2\,d\omega \right) - s^3 \frac{2}{3} \rho_s^{-1} \int_{\Sph^2} P^3 \,d\omega + O(s^4)
\end{align*}
and
\begin{align*}
D[\Omega_{R_s+\chi_s}] & = \frac{(4\pi)^2}{15} \rho_s^5
+ s^2 \rho_s^3 \left( - \frac{2\pi}{3} \int_{\Sph^2} P^2\,d\omega + \frac12 \iint_{\Sph^2\times\Sph^2} \frac{P(\omega)\,P(\omega')}{|\omega-\omega'|}\,d\omega\,d\omega' \right) \\
& \quad + s^3 \rho_s^2 \left( - \frac{7\pi}{9} \int_{\Sph^2} P^3\,d\omega + \frac{3}{4} \iint_{\Sph^2\times\Sph^2} \frac{P(\omega)^2\,P(\omega')}{|\omega-\omega'|}\,d\omega\,d\omega' \right) + O(s^4) \,.
\end{align*}
Thus,
\begin{align*}
& \mathcal I[\Omega_{R_s+\chi_s}] - \left( 4\pi \rho_s^2 + \frac{(4\pi)^2}{15} \rho_s^5 \right) \\
& \qquad = s^2 \left( \frac12 \int_{\Sph^2} (\nabla P)^2\,d\omega  - 6 \int_{\Sph^2} P^2\,d\omega + \frac{15}{4\pi}\iint_{\Sph^2\times\Sph^2} \frac{P(\omega)\,P(\omega')}{|\omega-\omega'|}\,d\omega\,d\omega' \right) \\
& \qquad\quad + s^2 \left(\rho_s^3 - \frac{30}{4\pi} \right) \left( - \frac{2\pi}{3} \int_{\Sph^2} P^2\,d\omega + \frac12 \iint_{\Sph^2\times\Sph^2} \frac{P(\omega)\,P(\omega')}{|\omega-\omega'|}\,d\omega\,d\omega' \right) \\
& \qquad\quad + s^3 \left( \!\left(- \frac{2}{3}\rho_s^{-1} - \frac{7\pi}{9}\rho_s^2\right) \int_{\Sph^2} P^3\,d\omega + \frac{3}{4}\rho_s^2  \iint_{\Sph^2\times\Sph^2} \!\!\!\frac{P(\omega)^2\,P(\omega')}{|\omega-\omega'|}\,d\omega\,d\omega' \right) + O(s^4).
\end{align*}
We now use the fact that $P$ is a spherical harmonic of degree $2$ and therefore it is an eigenfunction of $-\Delta$ and of the operator with integral kernel $|\omega-\omega'|^{-1}$ with eigenvalues $6$ and $4\pi/5$, respectively, see the proof of Proposition \ref{linearization}. This implies
$$
\frac12 \int_{\Sph^2} (\nabla P)^2\,d\omega  - 6 \int_{\Sph^2} P^2\,d\omega + \frac{15}{4\pi}\iint_{\Sph^2\times\Sph^2} \frac{P(\omega)\,P(\omega')}{|\omega-\omega'|}\,d\omega\,d\omega' = 0 \,,
$$
as well as
$$
\iint_{\Sph^2\times\Sph^2} \frac{P(\omega)\,P(\omega')}{|\omega-\omega'|}\,d\omega\,d\omega' = \frac{4\pi}{5} \int_{\Sph^2} P^2\,d\omega
$$
and
$$
\iint_{\Sph^2\times\Sph^2} \frac{P(\omega)^2\,P(\omega')}{|\omega-\omega'|}\,d\omega\,d\omega' = \frac{4\pi}{5} \int_{\Sph^2} P^3\,d\omega \,.
$$
Using Corollary \ref{volume} we obtain
$$
\rho_s^3 = R_s^3+\mathcal O(s^2) = \frac{30}{4\pi} - \frac{90}{7\cdot 4\pi} R_*^{-1} s + \mathcal O(s^2)
$$
and therefore
\begin{align*}
\mathcal I[\Omega_{R_s+\chi_s}] - \left( 4\pi \rho_s^2 + \frac{(4\pi)^2}{15} \rho_s^5 \right) & = s^3 R_*^{-1} \left( \frac{6}{7} \int_{\Sph^2} P^2\,d\omega  - 2 \int_{\Sph^2} P^3\,d\omega \right) + O(s^4).
\end{align*}
Finally, we compute
$$
\int_{\Sph^2} P^2\,d\omega = 2\pi \int_{-1}^1 P_2(t)^2 \,dt = \frac{4\pi}{5}
\qquad\text{and}\qquad
\int_{\Sph^2} P^3\,d\omega = 2\pi \int_{-1}^1 P_2(t)^3 \,dt = \frac{8\pi}{35}
$$
and obtain the formula in Theorem \ref{energy}.
\end{proof}



\bibliographystyle{amsalpha}

\end{document}